\newif\ifarxiv
\arxivtrue

\ifarxiv
\documentclass[11pt,a4paper]{article}
\usepackage[top=3cm, bottom=3.5cm, left=3cm, right=3cm]{geometry}
\else

\fi
\usepackage{url}
\usepackage{graphicx}

\graphicspath{{figures/}}

\ifarxiv
\title{Flip Distance Between Triangulations of a Planar Point Set is APX-Hard\thanks{Accepted Manuscript. The final version is available via \protect\url{http://dx.doi.org/10.1016/j.comgeo.2014.01.001}}}
\else
\fi

\author{Alexander Pilz%
\thanks{%
\ifarxiv
Institute for Software Technology, Graz University of Technology, Austria.
\texttt{apilz@ist.tugraz.at.}
\else
\fi
Recipient of a DOC-fellowship of the Austrian Academy of Sciences at the Institute for Software Technology, Graz University of Technology, Austria. Part of this work has been done while the author was visiting the Work Group Theoretical Computer Science at Freie Universit\"at Berlin, Germany.}}

\newcommand{\myqed}{\qed}

\date{December 30, 2013}

\ifarxiv
\usepackage{amsthm,amssymb,amsfonts,amsmath}

\renewcommand{\myqed}{}

\newtheorem{theorem}{Theorem}
\newtheorem{lemma}[theorem]{Lemma}
\newtheorem{corollary}[theorem]{Corollary}
\newtheorem{proposition}[theorem]{Proposition}
\newtheorem{definition}{Definition}
\newtheorem{problem}{Problem}

\else

\fi

\begin{document}

\maketitle

\begin{abstract}
In this work we consider triangulations of point sets in the Euclidean plane, i.e., maximal straight-line crossing-free graphs on a finite set of points.
Given a triangulation of a point set, an edge flip is the operation of removing one edge and adding another one, such that the resulting graph is again a triangulation.
Flips are a major way of locally transforming triangular meshes.
We show that, given a point set $S$ in the Euclidean plane and two triangulations $T_1$ and $T_2$ of $S$,
it is an APX-hard problem to minimize the number of edge flips to transform $T_1$ to $T_2$.
\end{abstract}

\section{Introduction}
Given a finite set $S$ of $n$ points in the Euclidean plane, a triangulation~$T$ of $S$ is a maximal straight-line crossing-free graph on $S$.
An \emph{edge flip} is the operation of removing an edge $e$ of~$T$ and adding a different edge $f$ such that the resulting graph~$\tilde T$ is again a triangulation of $S$.
This requires the two empty triangles incident to~$e$ to form a convex quadrilateral, which is the same as the one formed by the triangles incident to~$f$ in $\tilde T$.
The flip operation defines the graph~$\mathcal{G}$ of triangulations of~$S$, also called the \emph{flip graph} of~$S$.
For a given set $S$, the vertex set of~$\mathcal{G}$ is the set of all triangulations of~$S$.
Two vertices in $\mathcal{G}$ are adjacent if the corresponding triangulations can be transformed into each other by a single edge flip.
Lawson~\cite{lawson_connected} showed that~$\mathcal{G}$ is connected with diameter~$O(n^2)$ for any~$S$.
Hurtado, Noy, and Urrutia~\cite{hurtado_noy_urrutia} proved that this bound is tight.

Bose and Hurtado~\cite{survey} give an extensive survey on the flip operation.
Flips in triangulations are used for enumeration and as a local operation to generate meshes of good quality according to a predefined criterion.
For example, Lawson~\cite{lawson_delaunay} showed that one can always obtain the Delaunay triangulation after $O(n^2)$ locally improving flips.
The Delaunay triangulation optimizes several criteria.
Also, heuristic methods for improving other properties of triangular meshes may apply local optimization using flips in combination with techniques like simulated annealing.
See~\cite{bern_eppstein,hjelle} for information on the topic of mesh optimization.
Another reason for the continuing interest in flips in triangulations is the bijection between binary trees and triangulations of convex point sets.
There, a flip corresponds to a rotation in the binary tree.
Properties of the flip graph for convex point sets were studied in the landmark paper of Sleator, Tarjan, and Thurston~\cite{sleator}.
They show that, for $n > 12$, the flip distance between two triangulations is at most $2n-10$ and that, for sufficiently large~$n$, this bound is tight.
In a recent preprint, Pournin~\cite{pournin} shows a general lower bound construction for point sets in convex position, implying that the bound $2n-10$ is tight for all $n > 12$.

Interestingly, the flip distance problem is still open for point sets in convex position (or equivalently, convex polygons), regardless of the intensive investigation of that structure within the last 25 years.
The problem was apparently first considered by Culik and Wood~\cite{tree_similarity} in 1982.
Efforts were made in solving special cases and approximating the flip distance in polynomial time.
The results by Sleator et al.~\cite{sleator} lead to an algorithm to obtain an approximation of the flip distance within a factor of~$2$.
Li and Zhang~\cite{convex_approx_diagonals} give an algorithm that approximates the flip distance within a factor depending on the maximal vertex degree~$\Delta$ in source and target triangulation, obtaining a performance ratio bound of $2 - 2/(4(\Delta-3)(\Delta+4)+1)$.
Cleary and St.~John~\cite{fpt_convex_flips} show that the problem is fixed-parameter tractable in the flip distance.
Bose et al.~\cite{edge_labelled_triangulations} most recently considered edge-labeled triangulations, i.e., triangulations in which each edge has a distinct label, and, after a flip, the new edge gets the label of the removed edge.
For the flip distance problem, not only the edges but also their labels are given for the target triangulation.
They show that, in this setting, the flip distance can be $\Theta(n \log n)$ in the worst case, and gave an~$O(\log n)$-factor approximation algorithm for computing the flip distance between two edge-labeled triangulations.

For general point sets, Hanke, Ottmann, and Schuierer~\cite{edge_flipping_distance} show that the length of a shortest path between two triangulations in $\mathcal{G}$ (i.e., the \emph{flip distance}) can be bounded from above by the number of crossings between the edges of the two triangulations.
Eppstein~\cite{eppstein} gives a polynomial-time algorithm for computing a lower bound; note that the point sets for which Eppstein's result is tight must not contain empty convex 5-gons.
This property requires that more than two points are placed on a common line if the set has 10 or more points (see, e.g.,~\cite{empty5gon}).
Throughout this paper, we make the common assumption that $S$ is in general position, i.e., that no three points are collinear.

Despite these results, the complexity of determining the flip distance between two triangulations has been unknown.
Our main result is that the problem is APX-hard, which sheds light on a ``fundamental open issue''~\cite{survey} in the study of flip graphs.
Finding a polynomial-time algorithm to determine the flip distance has been addressed as an open problem by Hanke et al.~\cite{edge_flipping_distance} already in 1996, and, most recently, in a monograph by Devadoss and O'Rourke~\cite[p.~71]{devadoss}.
APX-hardness of the problem implies that no polynomial-time approximation scheme (PTAS) exists (i.e., there is no polynomial-time algorithm that approximates the flip distance by a ratio of at most $1+\epsilon$ for every constant $\epsilon > 0$), unless $\textsc{P}=\textsc{NP}$.
Most recently, NP-completeness of the problem has simultaneously and independently been shown by Lubiw and Pathak~\cite{lubiw}.
However, their reduction is from the \textsc{Planar Cubic Vertex Cover} problem, for which a PTAS exists~\cite{planar_ptas, baker_ptas} (see also~\cite[p.~369]{apx_book}), and the reduction can therefore not be adapted directly to show APX-hardness.
For triangulations of simple polygons, Aichholzer, Mulzer, and Pilz~\cite{poly_hard} recently showed that the corresponding problem is NP-complete.
\ifarxiv
A previous preprint version of this paper ({\tt arXiv:1206.3179v1}) only showed NP-completeness of the corresponding decision problem.
\fi

Clearly, the flip distance problem is an NP optimization problem.
Our reduction is from the well-known \textsc{Minimum Vertex Cover} problem.
In the next section, we show certain properties of triangulations of a class of point sets called double chains, which will be subsets in our construction.
In Section~\ref{sec_reduction}, we present the gadgets used in our reduction and analyze the construction.
In that section, we only present a rough overview on how the points of the set are placed, a more detailed description of how to calculate their coordinates is given in the appendix.

\section{Double-Chain Constructions}
\label{sec_double_chain}
A main ingredient of our reduction will be gadgets consisting of subsets that, being considered on their own, would require a number of flips that is quadratic in their size.

\subsection{A Single Double Chain}

We use definitions similar to~\cite{poly_hard}.
See \figurename~\ref{fig_dc}.
A \emph{double chain} $D$ is a point set of $2n$ points, $n$ on the \emph{upper chain} and~$n$ on the \emph{lower chain}.
Let these points be $\langle u_1, \dots, u_n \rangle$ and $\langle l_1, \dots, l_n \rangle$, respectively, ordered from left to right.
Any point on one chain sees every point of~$D$ on the convex hull boundary of the other chain (i.e., the interior of the straight line segment between these two points does not intersect the convex hulls of the two chains), and any quadrilateral formed by three points of one chain and one point of the other chain is non-convex.
Hurtado, Noy, and Urrutia~\cite{hurtado_noy_urrutia} show that the flip graph of the double chain has quadratic diameter.
Let $P_D$ be the polygon $\langle l_1, \dots, l_n, u_n, \dots, u_1 \rangle$.
The edges $u_i u_{i+1}$ and $l_i l_{i+1}$ for $1 \leq i < n$ have to be part of every triangulation of $D$ since there does not exist a straight-line segment between two points of $D$ that crosses any of them (such edges are called \emph{unavoidable}).
Therefore, we only need to consider the triangulation inside $P_D$ for the following result.

\begin{figure}
\centering
\includegraphics{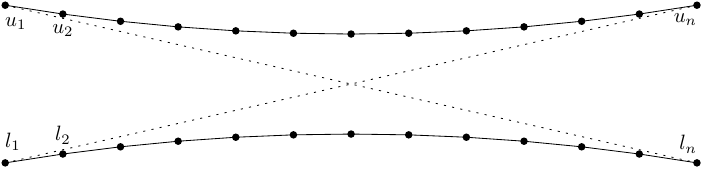}
\caption{A double chain.
The points are divided in an upper and lower chain, each chain being in convex position in a way that every point of the lower chain ``sees'' every vertex of the convex hull of the upper chain, and vice-versa.}
\label{fig_dc}
\end{figure}

\begin{theorem}[Hurtado, Noy, Urrutia]\label{thm_dc}
Consider any triangulation $T_1$ of $D$ where $u_1$ is adjacent to each of $l_1,\dots,l_n$, and any other triangulation~$T_2$, where $l_1$ is adjacent to $u_1,\dots,u_n$.
The flip distance between $T_1$ and $T_2$ is at least~\mbox{$(n-1)^2$}.
\end{theorem}

\begin{figure}
\centering
\includegraphics{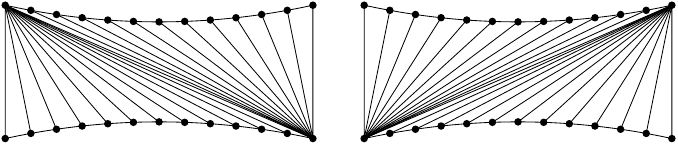}
\caption{Two (partial) triangulations of the double chain with a flip distance of at least~\mbox{$(n-1)^2$}.}
\label{fig_dc_triangulations}
\end{figure}

See \figurename~\ref{fig_dc_triangulations} for the relevant parts of the two triangulations.
In their proof, Hurtado et al.~\cite{hurtado_noy_urrutia} label the triangles inside~$P_D$ that have two points on the upper chain with \texttt{1} and the ones with two points on the lower chain with \texttt{0}.
Consider a horizontal line~$\ell$ that separates the two chains.
The triangles crossed by~$\ell$ define, from left to right, a sequence $\sigma$ of $(n-1)$ elements labeled \texttt{0} and $(n-1)$ elements labeled \texttt{1}, see~\figurename~\ref{fig_dc_01}.
Note that there are no triangles of a third type stabbed by~$\ell$.
Further note that we do not care about the triangulation of the convex hull of either chain; the lower bound on the flip distance stems from the part stabbed by~$\ell$.
It is easy to see that only an edge adjacent to two differently labeled triangles can be flipped in the stabbed part.
This corresponds to exchanging an adjacent pair of \texttt{0} and \texttt{1}.
Flipping the first triangulation to the second one corresponds to transforming the sequence $\sigma_1 = \langle (\texttt{0})^{n-1} (\texttt{1})^{n-1} \rangle$ to $\sigma_2 = \langle (\texttt{1})^{n-1} (\texttt{0})^{n-1} \rangle$, which leads to the desired bound.
We call these two triangulations (shown in \figurename~\ref{fig_dc_triangulations}) the \emph{extreme triangulations of $D$}.

\begin{figure}
\centering
\includegraphics{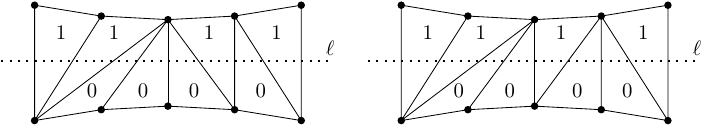}
\caption{An illustration of the labeling argument for the lower bound.
By the flip, the sequence changes from $\langle 11000101 \rangle$ to $\langle 11001001 \rangle$.}
\label{fig_dc_01}
\end{figure}

Our next step will be to gain more insight into the way the flip graph is altered by the addition of points.
For the following definition refer to \figurename~\ref{fig_dc_wedges}.

\begin{figure}
\centering
\includegraphics{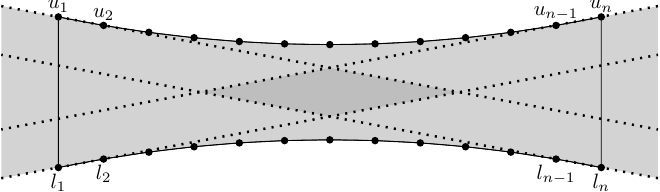}
\caption{The polygon~$P_D$ (bounded by solid lines) and the hourglass~$H_D$ (gray) of a double chain~$D$. The diamond-shaped flip-kernel can be stretched arbitrarily by flattening the bend of the chains.}
\label{fig_dc_wedges}
\end{figure}

\begin{definition}\label{def_dc}
Let $D$ be a double chain of $2n$ points, and consider the convex hulls of the upper and the lower chain.
Let $H_D$ be the continuous set of points such that for any point $p \in H_D$ there exist some $i,j$, $2 \leq i,j \leq n-1,$ with the triangle $p u_i l_j$ being interior-disjoint with the convex hulls of the upper and lower chain.
We call~$H_D$ the \emph{hourglass} of the double chain.
The \emph{flip-kernel of a double chain}~$D$ is the continuous set of points such that, for all $i,j$, $1 \leq i,j \leq n,$ and every point $p$ in the flip-kernel, the segments $p u_i$ and $p l_j$ are both interior-disjoint with the convex hulls of the upper and lower chain.%
\footnote{Note that the flip-kernel of $D$ may not be completely inside the polygon $P_D$ (but no point in the flip-kernel is outside the hourglass of~$D$).
This is in contrast to the common use of the term ``kernel'' in visibility problems for polygons.}
\end{definition}

Observe that the flip-kernel is the intersection of the open half-planes below $u_1 u_2$ and $u_{n-1} u_n$, as well as above $l_1 l_2$ and $l_{n-1} l_n$.
The hourglass is an unbounded region defined by the edges of $P_D$ and the rays defined by the first and the last vertex pair of each chain.

Let us add a point~$v$ inside the flip-kernel of $D$.
From any triangulation of the resulting set~$D \cup \{v\}$, we can flip the edges between the chains such that they are incident to $v$.
Reaching this canonical triangulation only requires a linear number of flips.
This fact is well-known folklore, see, e.g.,~\cite{problemas} for a printed description.
Consider the case where $v$ is placed outside $P_D$ but inside the flip-kernel of~$D$ (observe that the flip-kernel can be stretched by flattening the bend of the chains).
Add edges from $v$ to $u_1$ and $l_1$ to again have a triangulation, as shown in \figurename~\ref{fig_dc_steiner}.
Then, for flipping all possible edges to be incident to $v$, we need at most $2n-2$ flips.

In the remainder of this section, we will prove the following result, which shows that the quadratic lower bound holds if no point inside the hourglass is used to shorten the flip sequence.

\begin{figure}
\centering
\includegraphics{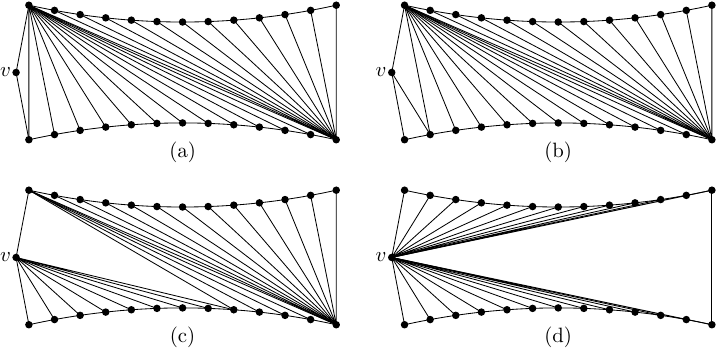}
\caption{An extra point~$v$ in the flip-kernel of~$D$ allows flipping one triangulation of~$P_D$~(a) to the other in $4n-4$ flips.
Note that an edge common to source and target triangulation is temporarily flipped~(b).}
\label{fig_dc_steiner}
\end{figure}

\begin{proposition}\label{prop_no_outer}
Let $D$ be a double chain of~$2n$~points and let $S \subset \mathbb{R}^2 \setminus P_D$ be a finite point set.
Let $T_1$ and $T_2$ be two triangulations of $S \cup D$ such that~$P_D$ is triangulated with one extreme triangulation of~$D$ in~$T_1$ and with the other extreme triangulation of~$D$ in~$T_2$.
Further, let $\sigma$ be a flip sequence from $T_1$ to $T_2$.
Assume that, throughout $\sigma$, no edge incident to a point of $S \cap H_D$ intersects the interior of~$P_D$.
Then $|\sigma| \geq (n-1)^2$.
\end{proposition}

In order to prove the proposition, we consider a mapping~$L$ from the set of triangulations of $S \cup D$ in $\sigma$ to the set of triangulations of the polygon~$P_D$.
When flipping an edge in a triangulation $T$ of $\sigma$, at most one edge is flipped in the corresponding triangulation~$L(T)$ of $P_D$.
Observe throughout the description that, informally, the mapping corresponds to continuously introducing the edges of the chains along the arrows drawn in \figurename~\ref{fig_dc_local}, while continuously sliding the edges of $T$ accordingly.

Consider any triangulation $T$ of $S \cup D$ in $\sigma$.
If all edges of $P_D$ are present, $L(T)$ equals the triangulation of $P_D$ in $T$ (note that this is also the case for the triangulations $T_1$ and $T_2$ of Theorem~\ref{thm_dc}).
Otherwise, consider the following construction (see \figurename~\ref{fig_dc_local} for an example).
For any edge $e$ of $T$ that intersects the hourglass of~$D$ and does not have any endpoint in the interior of the hourglass, we draw an edge $e'$ of $L(T)$ in the following way.
If one of the endpoints of $e$ is on a vertex of~$P_D$, then also one endpoint of~$e'$ is on that vertex.
If~$e$ passes through an edge $u_i u_{i+1}$ or~$l_j l_{j+1}$, then the corresponding upper or lower endpoint of~$e'$ is set to $u_{i+1}$ or~$l_{j+1}$, respectively.
If~$e$ passes through one of the rays defining the hourglass, then the corresponding endpoint of~$e'$ is mapped to the endpoint of the chain defining the ray;
for example, if $e'$ passes through the ray through~$u_1$ (starting at $u_2$) but not through the edge $u_1 u_2$, then the upper endpoint of~$e'$ is placed at $u_1$, such that~$e'$ is contained in $P_D$.
If an edge of~$T$ does not intersect the hourglass of~$D$ or has an endpoint in the interior of the hourglass, it is ignored by the mapping.

Let  $T' = L(T)$ be the graph induced by the new edges, and let the edges of~$T$ that pass through the hourglass but do not have an endpoint in $D$ be called \emph{wide} edges.
We call the construction $T'$ the \emph{local triangulation} of~$D$ when $T$ is clear from the context.
The following lemmata show that $T'$ actually is a triangulation of~$P_D$.

\begin{figure}
\centering
\includegraphics{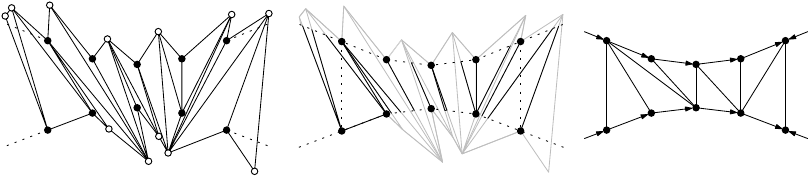}
\caption{Mapping a triangulation to a local triangulation of a double chain.
To the left, all triangles intersecting the hourglass of~$D$ are shown, the points of~$S$ are white.
Visually, one can think of ``cutting'' the edges at the boundary of the hourglass (middle) and moving (and merging) the endpoints to the next point~(right).}
\label{fig_dc_local}
\end{figure}

\begin{lemma}\label{lem_no_wide}
For every wide edge $e \in T$ that is mapped to $e' \in T'$, there is a different edge $\tilde e \in T$ that is also mapped to $e'$ and that has an endpoint $p \in D$.
\end{lemma}
\begin{proof}
Let $e'$ be $u_i l_j$.
Consider first the case where both endpoints of $e'$ are on the same side of (the directed line supporting)~$e$.
Consider the empty triangle $t$ of $T$ incident to $e$ that has its apex $a$ on the same side of~$e$ as~$e'$.
If~$a$ is outside the hourglass of~$D$, then another wide edge~$f$ of $t$ is also mapped to~$e'$.
In that case we continue the argument with $f$, as $e$ and $f$ are both mapped to the same edge.
If $a$ is not outside the hourglass, then~$a$ equals either $u_i$ or $l_j$, as otherwise $t$ would contain one of them (recall that no edge of $t$ is incident to a point of $S$ inside the hourglass).
Hence, one of the edges of $t$ incident to $a$ is also mapped to $e'$.

For the case where the two endpoints of $e'$ are on different sides of $e$ (i.e., one of the endpoints of $e'$ is $u_n$ or $l_n$), the argument is almost the same.
Without loss of generality, let $i=n$ and~$l_j$ be to the right of $e$ (note that $j$ may be $n$).
Therefore, $u_n$ is to the left of $e$.
Again, consider the empty triangle $t$ of $T$ incident to $e$ with apex $a$ to the right of $e$.
Again, if~$a$ is outside the hourglass of $D$, there is another wide edge $f$ of $t$ that is also mapped to $e'$.
If $a$ is not outside the hourglass, then $a = l_j$; this follows from the construction of $D$ and the fact that the lower endpoint of $e$ is outside the hourglass.
Hence, an edge of $t$ incident to $a$ is also mapped to $e'$.
\myqed
\end{proof}

\begin{lemma}\label{lem_inner_degree}
Every point $p$ of $D$ is incident to at least one edge $e$ of $T$ such that~$e$ disconnects the hourglass of~$D$.
\end{lemma}
\begin{proof}
This follows directly from the construction of $D$.
Suppose there is no such edge, and recall that there is also no edge incident to a point in the interior of the hourglass. Then there is an angle larger than~$\pi$ incident to~$p$, and the wedge defined by this angle contains points.
This contradicts the fact that $T$ is a triangulation.
\myqed
\end{proof}

\begin{lemma}\label{lem_local_triangulation}
$T'$ is a triangulation of $P_D$.
\end{lemma}
\begin{proof}
We have to prove that $T'$ is crossing-free and maximal in $P_D$.

Lemma~\ref{lem_no_wide} allows us to only consider non-wide edges.
With all relevant remaining edges of $T$ being incident to a point in $D$, the fact that $T'$ is crossing-free follows from $T$ being crossing-free, as the mapping only ``moves'' the endpoints of the edges of~$T$ to the next point of~$D$.

If $T'$ were not maximal, there would exist a quadrilateral $q$ inside $P_D$ that is spanned by points of~$D$ and whose interior does not intersect any edge.
If $q$ is not convex, this would mean that no edge of $T$ is incident to the reflex vertex of the quadrilateral.
But this cannot happen due to Lemma~\ref{lem_inner_degree} (an edge at that vertex in~$T$ that dissects the hourglass is mapped to an edge with the same property).
If~$q$ is convex, it is of the form $l_i l_{i+1} u_{j+1} u_j$.
If an edge of $T$ would have passed through the side~$l_i u_j$, the quadrilateral would not be empty of edges.
Hence, $l_i u_j$ must have been a part of $T$.
See \figurename~\ref{fig_no_quadrilateral}.
Since there are points to the right of the edge $l_i u_j$, there has to be a triangle of~$T$ adjacent to $l_i u_j$ having its third vertex to the right of that edge.
If the third vertex of the triangle is to the right of $l_i l_{i+1}$ or to the left of $u_j u_{j+1}$, one side of the triangle is mapped to a diagonal of~$q$ or the triangle would contain $l_{i+1}$ or $u_{j+1}$.
However, if the third vertex of the triangle is to the left of $l_i l_{i+1}$ and to the right of $u_j u_{j+1}$, it is inside the hourglass of the double chain.
Hence, there is no empty quadrilateral in $P_D$, which completes the proof.
\myqed
\end{proof}

\begin{figure}
\centering
\includegraphics{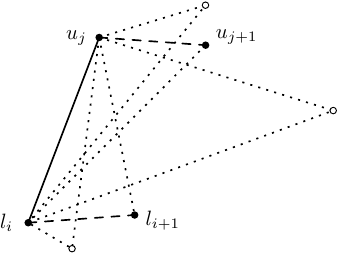}
\caption{The different possibilities for the triangle to the right of $l_i u_j$ in the triangulation~$T$.}
\label{fig_no_quadrilateral}
\end{figure}

At first sight, it might be conceivable that a flippable edge~$e$ of $T$ is mapped to a non-flippable edge $e'$ and that flipping $e$ to an edge $f$ results in an illegal flip of~$e'$ in the mapped triangulation~$L(T)$.
Recall, however, that the flip operation is defined as removing one edge of a triangulation and replacing it by another one.
Since the previous lemma proves that before and after the flip we have a triangulation given by mapping each edge, we know that if flipping $e$ changes~$L(T)$, then $e'$ must be flippable as well.
(Note, however, that if flipping $e$ does not change $L(T)$, there is another edge mapped to $e'$, and $e'$ may or may not be flippable; this will be discussed in Lemma~\ref{lem_mapping}.)

Since any flip in $T$ results in at most one edge being flipped in $T'$, the lower bound construction holds:
a shorter flip sequence with points outside the hourglass would immediately imply a shorter flip sequence between $T_1$ and $T_2$ in the proof of Theorem~\ref{thm_dc}.
This completes the proof of Proposition~\ref{prop_no_outer}.

\subsection{Multiple Double Chains}
Proposition~\ref{prop_no_outer} is, however, of little use when we try to construct a point set that contains many double chains and try to argue that the flip distance between two triangulations of the set is bounded by the sum of the distances between the local triangulations of these double chains.
One could imagine that a flip in the overall triangulation leads to changes in the local triangulations of several double chains.
In this section, we prove that this is not possible. %
Keep in mind that it is a necessary condition that, for any double chain~$D$, all other double chains are outside the hourglass of~$D$ and their polygons do not intersect.

\begin{lemma}\label{lem_mapping}
Let $e$ be a flippable edge of any triangulation $T$ of $D \cup S$ that is mapped to the edge $e'$ in the corresponding local triangulation~$T'$ of a double chain $D$.
Then flipping~$e$ changes the local triangulation only if no other edge is mapped to $e'$.
\end{lemma}
\begin{proof}
Suppose $e$ is not the only edge mapped to $e'$.
If we remove $e$ from~$T$, the graph on $D$ defined by the mapping is still the local triangulation~$T'$.
If we add the new edge $f$ after the removal of $e$, $f$ must also be mapped to some existing edge $f'$ in $T'$ (which might not be~$e'$) or is not mapped at all, as otherwise $T'$ would not be a triangulation. 
\myqed
\end{proof}

Note that because of Lemma~\ref{lem_mapping}, flipping an edge that is wide for a double chain does not change the local triangulation of that double chain.
Therefore, a flip can only change at most four local triangulations.
Actually, we can prove the following more accurate result.

\begin{lemma}\label{lem_no_common_flip}
Let $D_1$ and $D_2$ be two double chains in a point set $S$.
If each of $D_1$ and $D_2$ is outside the hourglass of the other and $P_{D_1} \cap P_{D_2} = \emptyset$, each flip in a triangulation of $S$ affects at most one of the two local triangulations.
\end{lemma}
\begin{proof}
If the flipped edge $e$ or its replacement $f$ do not both have an endpoint in the same double chain $D$, then at least one of $e$ or $f$ either does not dissect the corresponding hourglass or is a wide edge of $D$.
It follows from Lemma~\ref{lem_mapping} that such a flip does not influence the local triangulation of~$D$.
Hence, in the only remaining case there is a quadrilateral that has two adjacent points in $D_\mathrm{1}$ and two adjacent points in $D_\mathrm{2}$ and contains a flippable edge.
Let the quadrilateral be $abcd$.
Without loss of generality, let $a$ and $b$ be part of $D_1$ and $e = ac$.
See \figurename~\ref{fig_no_double_flip}.
Suppose, for the sake of contradiction, that we flip the edge $ac$ and the flip changes the local triangulation of $D_\mathrm{1}$.
Then $ac$ has to dissect the hourglass of~$D_\mathrm{1}$.
Then, however, $ad$, too, dissects the hourglass and crosses the same edge of $P_{D_\mathrm{1}}$ as~$ac$ (since the triangle $acd$ is empty).
Hence, $ac$ and $ad$ are mapped to the same edge in the local triangulation of $D_\mathrm{1}$, a contradiction due to Lemma~\ref{lem_mapping}.
\myqed
\end{proof}

\begin{figure}
\centering
\includegraphics{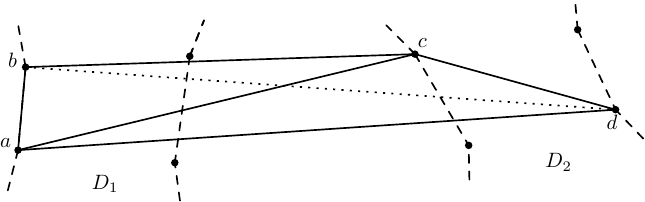}
\caption{An example illustrating why a flip cannot affect more than one local triangulation.
The edge~$ac$ is mapped to the same edge as~$ad$ in the local triangulation of~$D_1$.}
\label{fig_no_double_flip}
\end{figure}

\begin{corollary}\label{cor_distance}
If a point set consists of $m$ double chains, each of size $2n$, and for every double chain all other points are outside its hourglass, then the flip graph diameter of the whole set is in $\Omega(mn^2)$.
\end{corollary}

\section{The Reduction}
\label{sec_reduction}
Now we have gathered enough knowledge about double chains as sub-configurations in order to use them as the main building blocks in a reduction.
We reduce from \textsc{Minimum Vertex Cover}, which is known to be APX-complete~\cite{vertex_cover_apx}.\footnote{A previous version of this paper used a reduction from \textsc{Minimum Vertex Cover} on 3-regular graphs, which is also known to be APX-complete~\cite{cubic_apx}. However, as pointed out by an anonymous referee, reducing from the general version gives a better lower bound on the performance ratio without any substantial changes to the reduction.} %

\begin{problem}[\textsc{Minimum Vertex Cover}]
Given a simple graph $G = (V, E)$ with $n = |V|$, choose a set $C \subset V$ such that every edge in $E$ has at least one vertex in $C$ and such that $|C|$ is minimized.
\end{problem}

We follow the common approach of embedding the graph~$G$ and transforming its elements to geometric gadgets.
The gadgets consist of points together with the corresponding edges in the source triangulation $T_1$ and in the target triangulation~$T_2$.
We give the overall idea of how to embed the gadgets; for a detailed description on how to exactly place the points with rational coordinates having a representation bounded by a polynomial in the input size using polynomial time see the appendix.

\subsection{Gadgets}
Given a graph $G = (V, E)$ for which we have to solve the \textsc{Minimum Vertex Cover} problem, with $n = |V|$ and $m = |E|$, we place the elements of $V$ as the vertices of a convex $n$-gon and draw the straight-line edges between them (where the edges will not be part of the final construction).
Hence, we can consider~$G$ being a geometric graph in the remainder of this section.
For each edge~$e$ mark a point $c_e \in e$ that is not on a crossing.
Let $\vec t$ be a vector perpendicular to~$e$ of sufficiently small length (which will be specified in the appendix).
Make two copies of $e$ and translate them by $\vec t$ and~$-\vec t$, respectively, to obtain the \emph{tunnel} of the edge, i.e., the quadrilateral defined by the two copies of~$e$.
Then slightly ``bend'' the copies towards the (geometric) midpoint of $e$ to obtain two circular arcs $A_e$ and~$A'_e$.
The endpoints of the original edge~$e$ have to see any point on $A_e$ and~$A_e'$.
See \figurename~\ref{fig_k33}.

\begin{figure}
\centering
\includegraphics[width=\textwidth]{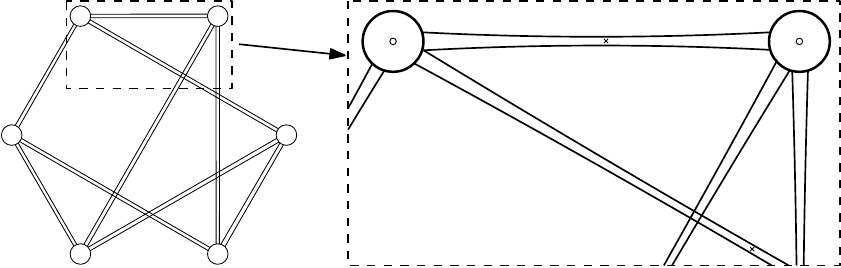}
\caption{An embedding of a graph with the (almost straight) circular arcs at each edge ending at a fixed distance around each vertex.}
\label{fig_k33}
\end{figure}

\subsubsection{Edge Cores}
Instances of the double chain are the main ingredient in our reduction.
They are contained in the gadgets representing the edges of~$G$.
See~\figurename~\ref{fig_edge_center} for an illustration of the construction.
Let $e$ be a straight-line edge of $G$, drawn between the points $v$ and $v'$.
In a close neighborhood of~$c_e$, place a double chain $D_e$, the \emph{edge core}, of $2d$ points (we will fix the value of~$d$ later) along $A_e$ and~$A'_e$ such that the two chains are separated by the supporting line of~$e$.
Note that the endpoints $v$ and $v'$ of $e$ are the only points that are not outside the hourglass of $D_e$, and they are also in the flip-kernel of~$D_e$ (remember that $A_e$ and $A_e'$ can be chosen sufficiently flat).
The edge cores are the only gadgets that have different edges in the source and in the target triangulation.
Draw the edges that define the polygon $P_{D_e}$ in both $T_1$ and~$T_2$.
Then triangulate the interior of $P_{D_e}$ with one extreme triangulation of $D_e$ in~$T_1$ and with the other extreme triangulation in $T_2$.
We refer to the process of flipping edges that are incident to an edge core as \emph{transforming an edge core}.

\begin{figure}
\centering
\includegraphics{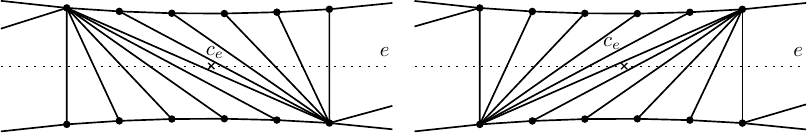}
\caption{The double chain at the center of an edge with the source and the target triangulation.}
\label{fig_edge_center}
\end{figure}

\subsubsection{Crossings}
If two straight-line edges $e$ and $f$ of $G$ cross, also their corresponding circular arcs cross.
The four circular arcs define a region bounded by four pieces of the original arcs.
Place one point at each of the four crossings of the arcs (we will actually place the points not exactly on the crossings, but close, see the appendix).
In both source and target triangulation draw the edges connecting two points that are consecutive on any circular arc, which results in a crossing being represented by a convex quadrilateral, to which we add an arbitrary diagonal.
Note that the crossing gadgets do not overlap with the edge core gadgets, as the edge cores are placed in the neighborhood of $c_e$, which was chosen not to be at a crossing.

\subsubsection{Wirings}
\emph{Wirings} are gadgets that represent the elements of $V$.
See~\figurename~\ref{fig_wiring} for an illustration.
Consider any vertex~$v$ of $G$ and a small circle $C$ with $v$ in the embedding as its center.
This \emph{point} $v$ is part of the triangulated point set.
Place points on the crossings of $C$ with the arcs of the edges incident to~$v$ in~$G$.
Since the graph is embedded on a convex $n$-gon and due to the small length of the vector $\vec{t}$, these points occupy strictly less than half of~$C$.
This allows us to place two chains $L$ and $R$, each of $w-1$ points (the value of $w$ is to be defined later) on $C$ in a way that any line between one point of~$R$ and one point of~$L$ separates~$v$ from the remaining construction.
In both the source and target triangulation draw the edges between consecutive points on $C$.
Draw a zig-zag path through the points of $L$, $R$, and the first and last point where $C$ crosses the arcs of the edges (giving $2w$ points in total).
We call these edges the \emph{zig-zag edges} of the wiring.
Connect $v$ to the first point of~$L$ and to the first point of~$R$.
The remaining part may be triangulated arbitrarily.

\begin{figure}
\centering
\includegraphics{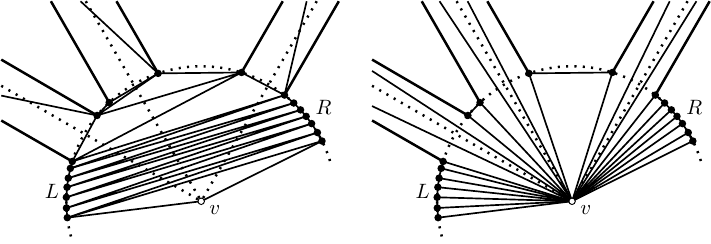}
\caption{Left: A wiring with its initial and final triangulation (solid).
Right: A triangulation that allows to quickly perform a transformation of the edge cores.
The parts of the auxiliary construction shown in \figurename~\ref{fig_k33} are dotted.
}
\label{fig_wiring}
\end{figure}

The remaining faces in the two plane graphs we obtained so far are triangulated arbitrarily, however in a way that the resulting triangulations $T_1$ and $T_2$ have the same edges except at the edge cores.

\subsection{Analysis}
\newcommand{\opt}{\ensuremath{\mathrm{opt}}}
The basic idea of the construction is that a flipping algorithm that gives the shortest flip distance or a good approximation of it has to choose which wirings to flip (requiring $4w-2$ flips each for flipping the zig-zag edges of a wiring away and back again) in order that the triangulation of an edge core can be transformed using the point in its flip-kernel at the chosen wiring.
Also, the at most~$4x+2$ edges between and at the crossings need to be flipped away.
We will fix the values of $w$ and $d$ to force this behavior of any flipping algorithm that uses fewer flips than a trivial upper bound.
Every edge of $G$ will be covered; using a vertex of~$G$ for covering corresponds to flipping the zig-zag edges in the corresponding wiring.

Let $v$ and $v'$ be any two adjacent vertices in $G$.
The exact number of edges in $T_1$ or~$T_2$ intersected by the segment $vv'$ in the drawing may differ with the choice of~$v$ and~$v'$ because (i) the number of crossings of each edge of~$G$ may differ, and (ii) the triangulation of the wiring gadget at the region where the edge gadgets enter it is not completely symmetric.
Let $x$ be the maximum number of crossings of a single edge in~$G$.
For every wiring, the number of edges that are intersected by the segment~$vv'$ in addition to the zig-zag edges is at most $2n-3$ (the remaining part is a $2n$-gon, see \figurename~\ref{fig_wiring}).
We denote the sum of these numbers over all wirings by~$\tau$; we have $\tau \in O(n^2)$.

The following lemma shows how to deduce a flip sequence in our construction from a vertex cover of size $k$.
Note that we do not claim that this is the optimum if $k$ is optimal.
\begin{lemma}\label{lem_flip_from_cover}
If there exists a vertex cover of size $k$ in $G$, then there exists a flip sequence between $T_1$ and $T_2$ of length at most
\begin{equation*}
\delta_k = 2(k(2w-1) + m(4x + 2d) + \tau) \enspace .
\end{equation*}
\end{lemma}
\begin{proof}
Let $C$ be a vertex cover of $G$ with $k = |C|$.
Let $v \in C$ be a vertex used to cover an edge.
We use $v$ to transform the edge cores of the adjacent edges in $G$ (if they have not already been transformed).
We need to flip all zig-zag edges in the wiring to $v$, which takes $2w-1$ flips.
Then we need at most $2n-3$ flips (counted by $\tau$) for the remaining wiring edges, as well as two further flips for the edges before the first crossing and two flips for the first crossing itself.
All in all, with this method we need up to $4x+2$ flips for the crossing gadgets to make the first edge of the edge core visible to~$v$.
Then, we need $2d-2$ flips to make the edges incident to~$v$ (see \figurename~\ref{fig_dc_steiner}~(d)).
Flipping in the desired way we need at most $\delta_k$ flips.
\end{proof}

On the other hand, a flip sequence should define a vertex cover.
For the following lemma, we fix
\[
w > \frac{c(m(4x+2d) + \tau) + 1}{2}
\]
for any constant $c > 1$; further, we choose $d$ such that $(d-1)^2 > \delta_n = 2(n(2w-1) + m(4x+2d) + \tau)$ (note that since the term to the right is linear in~$d$, such a value of~$d$ clearly exists and is polynomial in the problem size).
\begin{lemma}\label{lem_cover_from_flip}
If there exists a flip sequence between $T_1$ and $T_2$ of length $\delta$, then there exists a vertex cover of size at most
\begin{equation}\label{eqn_delta_to_k}
k = \left \lfloor \frac{\delta}{4w-2} \right \rfloor \enspace.
\end{equation}
In particular, for the flip distance $\delta_\opt$ between $T_1$ and $T_2$ and a minimum vertex cover of size $k_\opt$, we have 
\begin{equation}\label{eqn_delta_opt_to_k_opt}
k_\opt = \frac{\delta_\opt - R}{4w-2}
\end{equation}
for some positive $R < \frac{4w-2}{c}$.
\end{lemma}
\begin{proof}
We argue that the choice of $d$ forces an effective algorithm to flip the zig-zag edges of wirings (which corresponds to covering vertices), and that the choice of $w$ allows to transform the number of flips to the size of the corresponding vertex cover.

If $\delta \geq (d-1)^2$, then the choice of $d$ implies that $k \geq n$ in~(\ref{eqn_delta_to_k}), which trivially implies that the lemma is true in that case.
We therefore assume that $\delta < (d-1)^2$.
If, for any edge core, we do not use the corresponding central points~$v$ or~$v'$ of a wiring, we need at least $(d-1)^2$ flips due to Proposition~\ref{prop_no_outer}.
Now suppose that we want to transform an edge core~$D$ using a point~$v$.
Then we need to flip all zig-zag edges in the wiring to $v$ (as in the proof of Lemma~\ref{lem_flip_from_cover}), taking $2w-1$ flips.
Note that this is optimal since only one of the zig-zag edges can be removed with each flip.
The values of $d$ and~$w$ have been chosen in a way that flipping the edges of all wirings, crossings, and edge cores to the corresponding central point and back, as described, uses fewer flips than transforming one edge core, due to the bound of Lemma~\ref{lem_flip_from_cover}.
For any algorithm, this means that flipping all edges at wirings and crossings twice and transforming the edge cores with a point at the wiring is cheaper than transforming one edge core without a point at a wiring.
Due to Proposition~\ref{prop_no_outer} we know that we need a point at a wiring for each edge core to be transformed in fewer than $(d-1)^2$ flips, as, for each edge core, the points at the two wirings are the only ones inside the hourglass of the edge core.
Therefore, we know that the (optimal) flip distance $\delta_\mathrm{opt}$ is given by
\begin{equation}\label{eqn_delta_opt}
\delta_\mathrm{opt} = k_\opt(4w - 2) + R \text{ for some } R > 0 \enspace .
\end{equation}

Equation~(\ref{eqn_delta_opt}) shows how to deduce~$k_\opt$ from~$\delta_\mathrm{opt}$:
Lemma~\ref{lem_flip_from_cover} gives us an upper bound on the flip distance, and hence $R \leq 2(m(4x + 2d) + \tau)$.
Note that if $R < 4w-2$, the size of the minimum vertex cover can be calculated from the flip distance by
\begin{equation*}%
k_\opt = \left \lfloor \frac{\delta_\mathrm{opt}}{4w-2} \right \rfloor \enspace .
\end{equation*}
We actually require $cR < 4w-2$, for a given constant~$c > 1$ (which is used for the reasoning about approximation ratios later in this section).
This requirement can be fulfilled by choosing~$w$ under consideration of the bound $R \leq 2(m(4x + 2d) + \tau)$, i.e., such that $2(m(4x + 2d) + \tau) < (4w-2)/c$.
Thus, we have chosen $w$ such that, in an optimal flip sequence, flipping the zig-zag edges of one wiring needs more flips than $c$ times the number of all flips of edges not in a wiring.

No matter how well an algorithm performs, it has to flip the zig-zag edges of at least $k_\opt$ wirings when using less than $(d-1)^2$ flips, and Lemma~\ref{lem_flip_from_cover} tells us that $c$ times the number of flips of the edges not in a wiring are in total fewer than the number of the zig-zag edges flipped for one wiring when the algorithm is optimal.
\end{proof}

To show APX-hardness of the flip distance problem, we show that we have an AP-reduction~\cite[pp.~256--261]{apx_book} from \textsc{Minimum Vertex Cover} using the previous lemmata.
Let $k_\opt$ be the size of a minimum vertex cover for $G$ and $\delta_\opt$ be the flip distance between $T_1$ and $T_2$.
The \emph{performance ratio} of an approximate solution to a minimization problem is the value of the measure function applied to the approximation divided by the optimal value, e.g., $k/k_\opt$ for an approximate vertex cover of size~$k$.
See~\cite[pp.~257--258]{apx_book} for the following definition (note that~$r$ is a bound on the performance ratio of the approximate solution of the problem we reduce to, and that~$\alpha$ is a factor in the bound for the performance ratio of the solution to the initial problem).

\begin{definition}[AP-reduction]\label{def_ap_reduction}
Let $P_1$ and $P_2$ be two NP optimization problems.
$P_1$ is \emph{AP-reducible} to $P_2$ if two functions $f$ and $g$ and a constant $\alpha \geq 1$ exist such that:
\begin{enumerate}
 \item\label{item_is_instance} For any instance $X$ of $P_1$ and any rational $r > 1$, $f(X,r)$ is an instance of $P_2$.
 \item\label{item_has_solution} For any instance $X$ of $P_1$ and any rational $r > 1$, if there is a feasible solution of $X$, then there is a feasible solution of $f(X,r)$.
 \item\label{item_has_original_solution} For any instance $X$ of $P_1$ and any rational $r > 1$, and for any $Y$ that is a feasible solution of $f(X,r)$, $g(X,Y,r)$ is a feasible solution of $X$.
 \item\label{item_polynomial_time} $f$ and $g$ are computable by two algorithms whose running time is polynomial for any fixed rational $r$.
 \item\label{item_ratio_implied} For any instance $X$ of $P_1$ and any rational $r > 1$, and any feasible solution $Y$ for $f(X,r)$, a performance ratio of at most $r$ for $Y$ implies a performance ratio of at most $1 + \alpha (r-1)$ for $g(X,Y,r)$.
\end{enumerate}
\end{definition}
In our case, $f$ corresponds to the construction of the point set and the two triangulations.
Requirements \ref{item_is_instance} and~\ref{item_has_solution} follow from our construction.
A vertex cover can be extracted from a flip sequence~$Y$ from the zig-zag edges flipped at the wirings;
this corresponds to $g$, and requirement~\ref{item_has_original_solution} is therefore fulfilled.
Both $f$ and $g$ are polynomial-time algorithms, as demanded by requirement~\ref{item_polynomial_time} (the parameter~$r$ is actually not used by either of these two algorithms, but will be used in the analysis).

Intuitively, Lemmata~\ref{lem_flip_from_cover} and~\ref{lem_cover_from_flip} give evidence that the reduction described so far fulfills also requirement~\ref{item_ratio_implied} of Definition~\ref{def_ap_reduction}.
However, because of the remainder term~$R$, the performance ratio of an approximation of the flip distance does not directly give the performance ratio of the resulting approximate vertex cover;
we have to show that~$R$ was chosen small enough and therefore the performance ratio of the approximate vertex cover stays within the bounds required by Definition~\ref{def_ap_reduction}.
Let $\delta$ be an approximate solution for the flip distance such that $\delta \leq \delta_\opt r$.
Further, let $R'$ be the remainder produced by the floor function in~(\ref{eqn_delta_to_k}) of Lemma~\ref{lem_cover_from_flip}, that is, in the expression $k = \left \lfloor \frac{\delta}{4w-2} \right \rfloor$.
By Lemma~\ref{lem_cover_from_flip}, we get
\[
  k \leq \frac{\delta - R'}{4w-2} \leq \frac{\delta_\opt r - R'}{4w-2} \enspace .
\]
Let $R$ be the remainder term for the optimal solution $\delta_\opt$ as in~(\ref{eqn_delta_opt_to_k_opt}) of Lemma~\ref{lem_cover_from_flip}, that is, in the expression $k_\opt = \frac{\delta_\opt - R}{4w-2}$.
Then introducing the term $rR - rR$ in the numerator of the previous upper bound for~$k$ yields
\begin{equation}\label{eqn_k_ratio}
 k \leq r\frac{\delta_\opt - R}{4w-2} + \frac{rR - R'}{4w-2} = r k_\opt + \frac{rR - R'}{4w-2} \leq r k_\opt + \frac{rR }{4w-2} < r k_\opt + \frac{r}{c} \enspace ,
\end{equation}
where the equality and the last inequality are due to Lemma~\ref{lem_cover_from_flip}.
Let $\alpha = 4$ and $c=2$.
Suppose first that $r-1 = \epsilon \geq \frac{1}{2k_\opt + 1}$.
Then
\begin{equation}\label{eqn_large_approx}
 r k_\opt + \frac{r}{2} = k_\opt + \epsilon k_\opt + \frac{1}{2} + \frac{\epsilon}{2} = k_\opt + \alpha \epsilon k_\opt + \frac{1}{2} - \epsilon \left (3k_\opt - \frac{1}{2} \right ) \enspace . \enspace
\end{equation}
To get rid of the last part we use
\[
 \epsilon \left (3k_\opt - \frac{1}{2} \right) \geq \frac{3k_\opt-1/2}{2k_\opt+1} > \frac{1}{2} \enspace ,
\]
which, by (\ref{eqn_k_ratio}) and (\ref{eqn_large_approx}), implies
\[
 k \leq k_\opt + \alpha \epsilon k_\opt \enspace .
\]
On the other hand, suppose that $r-1 = \epsilon < \frac{1}{2k_\opt +1}$.
Then from (\ref{eqn_k_ratio}), we get
\begin{gather*}
 k < r k_\opt + \frac{r}{2} = k_\opt + \epsilon k_\opt + \frac{1}{2} + \frac{\epsilon}{2} = k_\opt + \epsilon\left (k_\opt + \frac{1}{2} \right ) + \frac{1}{2}\\
 < k_\opt + \frac{k_\opt + 1/2}{2 k_\opt + 1} + \frac{1}{2} = k_\opt + 1 \enspace .
\end{gather*}
Since the solutions to vertex cover are integers, this implies that $k = k_\opt$ and therefore $ k \leq k_\opt + \alpha \epsilon k_\opt$ holds.
Hence, in both cases $k / k_\opt \leq 1 + \alpha(r-1)$ and our reduction fulfills all properties of an AP-reduction from \textsc{Minimum Vertex Cover}.

\begin{theorem}
The problem of determining a shortest flip sequence between two triangulations of a point set is APX-hard.
\end{theorem}

\subsection{An Improved Bound on the Performance Ratio}
\label{sec_improved_bound}
The previous reduction did not use the performance ratio bound~$r$.
As pointed out by an anonymous referee, a different choice of $w$ actually allows to prove a better lower bound on the tractable performance ratios.
This reduction selects $c$ (the constant used in Lemma~\ref{lem_cover_from_flip}) according to~$r$ (recall that $r$ is considered a constant).
Hence, this is an example of a reduction that actually uses the bound~$r$ as a parameter.
It is known that approximating \textsc{Minimum Vertex Cover} by any constant factor less than $10\sqrt{5} - 21 \approx 1.36$ is NP-hard~\cite{dinur}, and, if the Unique Games Conjecture is true, even obtaining a performance ratio within any constant less than 2 is NP-hard~\cite{khot}.
However, there exist approximation algorithms achieving a ratio of $2 - o(1)$~\cite{hochbaum,karakostas}.

Let $b$ be the bound for the performance ratio that a polynomial-time algorithm can guarantee for \textsc{Minimum Vertex Cover} (note that $b$ is between $1.36$ and $2$, unless $\textsc{P}=\textsc{NP}$).
Suppose we can approximate the flip distance by a performance ratio less than $b - \varepsilon$ for some constant $\varepsilon$.
Due to (\ref{eqn_k_ratio}), we can guarantee a performance ratio of at most $(b-\varepsilon) + \frac{b-\varepsilon}{k_\opt c}$ for \textsc{Minimum Vertex Cover}.
Hence, if $\varepsilon > \frac{b-\varepsilon}{k_\opt c} + \varepsilon'$, then the performance ratio bound for \textsc{Minimum Vertex Cover} is better than $b-\varepsilon'$.
This is fulfilled for $c > \frac{b-\varepsilon}{\varepsilon - \varepsilon'}$.
In particular, this requires $\varepsilon = \kappa \varepsilon'$ for a constant $\kappa > 1$.
Note, however, that $\kappa$ cannot be 1.
The reason for this is that, in~(\ref{eqn_k_ratio}), $R'$ can be smaller than $R$.
For example, there may exist a 2-approximation for \textsc{Minimum Vertex Cover} for which the corresponding flip sequence is less than twice the optimum.
Still, we obtain the following result.

\begin{theorem}
For any given constant $\varepsilon > 0$, it is NP-hard to approximate the flip distance between two triangulations by a factor less than $10\sqrt{5} - 21 - \varepsilon$, and, if the Unique Games Conjecture is true, by a factor less than $2 - \varepsilon$.
\end{theorem}

\section{Conclusion}
In this paper, we showed that it is APX-hard to minimize the number of flips to transform two triangulations $T_1$ and $T_2$ of a point set $S$ into each other.
As a by-product, Corollary~\ref{cor_distance} revealed an interesting aspect on distances in the flip graph.

We are not aware of any constant-factor approximation of the flip distance.
For the upper bound given by Hanke et al.~\cite{edge_flipping_distance}, it is easy to construct examples (like the one in \figurename~\ref{fig_dc_steiner}) where the bound is quadratic while the flip distance is linear.

Given the recent NP-completeness result for simple polygons~\cite{poly_hard}, the main remaining open problem is the one for triangulations of convex point sets and its dual problem, the computation of the binary tree rotation distance~\cite{sleator}.

\ifarxiv
\paragraph{Acknowledgements.}
\fi
The author wants to express his gratitude to Oswin Aichholzer, Thomas Hackl, and Pedro Ramos, as well as anonymous referees for valuable suggestions on improving the presentation of the result.
In particular, one anonymous referee pointed out that a slight generalization of the reduction actually implies the result discussed in Section~\ref{sec_improved_bound}.
\appendix

\newcommand{\Q}{\ensuremath{\mathbb{Q}}}

\section{Calculation of the Coordinates}\label{apx_coordinates}
Section~\ref{sec_reduction} already contained a description of the gadgets we used in our reduction.
However, the validity of gadget-based reductions when proving NP- or APX-hardness for problems on point sets requires that the coordinates of the points used can be calculated in polynomial time.

The reader may have noticed that our high-level construction involves points placed at the crossing of circular arcs, which, in general, leads to irrational coordinates, even if the circular arcs are defined by rational points.
We will give a construction that slightly varies from the one described that uses only rational coordinates, with both the numerator and denominator bounded by a polynomial in the input size.

One way to strengthen the result is to show that the problem remains APX-hard for triangulations of point sets in general position.
The gadgets in our reduction do not make use of collinear points.
However, we did not explicitly mention how to avoid three points on a line when describing the construction.
In this appendix we give an explicit construction of the point set in general position, i.e., that no three points are collinear.

Note that the construction may not be ``economical'' in the sense that the construction may be possible with coordinates having a smaller binary representation.
We will always prefer constructions that are easy to prove.
We will place the points on and close to the unit disc (meaning that a coordinate will never exceed $1+\epsilon$, for some small $\epsilon > 0$);
therefore, we can specify the size of a coordinate in terms of the size of its denominator.

\subsection{Placing the Points of the Convex Polygon}

As a first step, we give a simple construction of a convex $n$-gon for placing the central points of the wiring gadgets with all vertex coordinates being rational and the denominators being in $O(n^{10})$.
Further, we want to assure that no three diagonals cross in the same point.
For doing so, we will first choose $n^5$ \emph{candidate points} on the unit circle and then select $n$ points out of them.

Rational points on the unit circle are known to be given by $\left (\frac{1-t^2}{t^2+1}, \frac{2t}{t^2+1} \right )$ with $t \in \Q$, see, e.g., \cite{canny}.
We define a sequence $K$ of candidate points with $t = i/n^5$ for the integers $1 \leq i \leq n^5$.
(For consistency with later parts and ease of presentation therein we choose the candidate points from the upper-right quadrant in counterclockwise ascending order; hence, the value of $t$ is between $0$ and $1$.)
Now we select $n$ points out of $K$ such that there are no three diagonals that cross at a single point.
We choose the first five points of our final set from the candidate points.
Suppose we have chosen $j \geq 5$ points such that no three diagonals cross at a single point.
We have $n^5 - j$ points in $K$ to choose the next point from.
Consider all~$\binom{j}{5}$ combinations of five points among the already chosen ones.
Each combination gives exactly five points on the unit circle that cannot be chosen, and none of these is among the $j$ already chosen candidate points.
Hence, we have $5 \binom{j}{5} + j$ ``forbidden'' points (which may not all be among the candidate points).
See \figurename~\ref{fig_app_vertices_on_circle}.
We have, however, $n^5 \geq j^5 > 5\binom{j}{5} + j$ candidate points to choose from, and therefore we for sure can choose point number $(j+1)$.
We denote this set of points by~$P_V$; the elements of $P_V$ are the points representing the vertices of the input graph.

\begin{figure}
\centering
\includegraphics{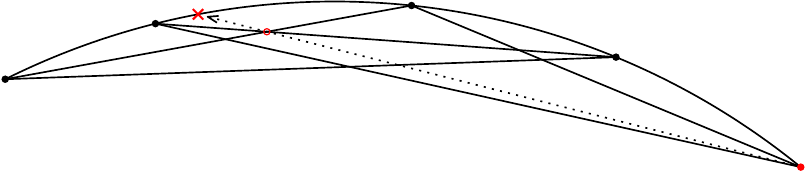}
\caption{Five points on the unit circle; a point at the (red) cross would introduce a supporting line through a crossing of two other supporting lines and is therefore forbidden.
The image is rotated for representational reasons, our method chooses all points from the upper-right quadrant.
}
\label{fig_app_vertices_on_circle}
\end{figure}

\begin{proposition}
A point set of $n$ points in convex position with all coordinates rational having their denominators in $O(n^{10})$ and no three diagonals crossing at the same point can be found in polynomial time.
\end{proposition}

Note that the facts that no three diagonals of the resulting $n$-gon cross and that the coordinates are bounded also give us a lower bound on the distance between intersection points and other diagonals, which we will use in the next part.

\subsection{A Sufficiently Small Value}
In this section, we will define four values $\delta_\mathrm{e}, \delta_\mathrm{v}, \delta_\mathrm{n},$ and $\delta_\mathrm{r}$ that will give sufficiently small upper bounds on the construction of the gadgets.
For any point~$p$, let $x_p$ and $y_p$ denote its $x$- and $y$-coordinate, respectively.

For the definition of $\delta_\mathrm{e}$, find the minimum squared distance from each of the $\binom{n}{4}$ crossings of the diagonals of the $n$-gon to the diagonals not involved in the corresponding crossing.
Let the actual distance be~$\delta_\mathrm{e}$.
Since the squared distance~$\delta_\mathrm{e}^2$, is given by $(x_a - x_b)^2 + (y_a - y_b)^2$ between two points $a$ and $b$, we can set $\delta_\mathrm{e}' = |x_a - x_b|$ to obtain a ``small'', rational and positive distance $\delta_\mathrm{e}' \leq \delta_\mathrm{e}$ (at least one of the horizontal or vertical distances is non-zero, in particular, up to here no two points can have the same $x$- or $y$-coordinate).
When we construct the tunnels that are formed around an edge of the drawing of the input graph, we can choose, say, $\delta_\mathrm{e}'/3$ as an upper bound for the distance between the edge and the edges defining the tunnel.
Then the intersection of any three tunnels is always empty.
(Our actual tunnels will be even narrower.)

The vertex gadgets used ``small'' circles around each point in $P_V$.
Let $u, v, w$ be a triplet of consecutive vertices on the $n$-gon defined by $P_V$.
Let $\delta_\mathrm{v}^2$ denote the smallest squared distance between $v$ and the line through $u$ and $w$ for every choice of the triplet.
As with the tunnels, we can choose a rational $\delta_\mathrm{v}' \leq \delta_\mathrm{v}$ by choosing only the horizontal or vertical distance between $v$ and the closest point on the supporting line of $u$ and $w$.

Again, let $v$ be a vertex on the $n$-gon.
Let $\ell_{v}$ be the line through~$v$ that is perpendicular to the line~$ov$, where~$o$ is the origin.
Consider the distances from $u$ and $w$ to $\ell_{v}$.
Let $\delta_\mathrm{n}^2$ be the smallest squared distance for all choices of $v$ (and corresponding $u$ and $w$), and choose a rational $\delta_\mathrm{n}' \leq \delta_\mathrm{n}$ as before.
Further, let $\delta_\mathrm{r}$ be the smallest horizontal or vertical distance between two points in $P_V$ (which is non-zero by construction).
See \figurename~\ref{fig_app_something_small}.
We define $\delta = \min\{\delta_\mathrm{e}'/3,\delta_\mathrm{v}',\delta_\mathrm{n}',\delta_\mathrm{r}\}$.
If we now choose the radius of the cycle centered at each vertex by $r_V = \delta/6$, then no two circles intersect (there is actually a distance of at least $4r_V$ between two circles), and each circle only intersects the edges of the input graph that are incident to the vertex it is centered at.
Further, no circle intersects the convex hull of two other circles.

\begin{figure}
\centering
\includegraphics{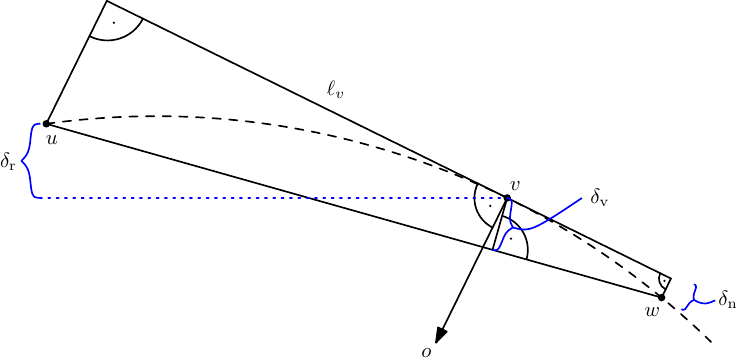}
\caption{Construction to obtain bounds for $\delta$.}
\label{fig_app_something_small}
\end{figure}

\subsection{Tunnel Construction}
For each edge $e$ of the input graph connecting two vertices $v$ and $w$, we now give the construction of the tunnels.
Let $C_v$ and $C_w$ be the circles around $v$ and $w$, respectively.
The tunnel for the edge between $v$ and $w$ is given by two segments, each having one endpoint on $C_v$ and one endpoint on~$C_w$.
We want to get rational points on $C_v$ and $C_w$.
Since these circles are not only defined by a rational center point, but also have a rational radius, the problem boils down to finding a rational point on the unit circle, or, equivalently, a (possibly irrational) angle $\alpha$ such that $\sin(\alpha)$ and $\cos(\alpha)$ are rational, within some interval given by quadratic irrationals.
Sines with this property are called \emph{rational sines}, and correspond with the parametrization of the unit circle that we already used before.
Canny, Donald and Ressler~\cite{canny} give an algorithm for finding a rational sine for a parameter $t = p/q$ such that $|p/q - x| < \epsilon$, for given $x$ and~$\epsilon$ (we will use an extended method for non-rational radii later).
Their algorithm gives a denominator~$q$ in $O(1/\epsilon)$, and the running time is polynomial in $q$.
However, the input $x$ is an approximation as well, and their goal is to get rational sines with small binary representation.
Our angle intervals, however, are given by rational points and their relative position to the circle center.
For finding a point within this interval, the Farey approximation as used by Canny et al.~\cite{canny} for $t = p/q$ is sufficient and easy to apply for our setting, as we do not need an explicit approximation of the angle and the interval as input (this algorithm searches a point inside the interval in the fashion of binary search, computing the mediant~$\frac{a+c}{b+d}$ of two rational values~$\frac{a}{b}$ and $\frac{c}{d}$ in each step).
We, however, need an upper bound on the denominator $q$ derived from the points defining the angle.

Now we show how to use the results by Canny et al.~\cite{canny} for our needs.
Consider the unit circle and two points $a$ and $b$.
Let $\angle a$ and $\angle b$ be the polar angles of these points, and, without loss of generality, let $\angle a < \angle b$.
We describe only the case where both angles are within $[0, \dots, \pi/2]$, the other cases are similar (and can easily be distinguished);
in our setting we simply have to rotate the plane orthogonally.
To approximate an angle between $\angle a$ and $\angle b$ using a rational number $t$, we reason about the (possibly irrational) values $t_a$ and~$t_b$.
For $t_a$ and $\angle a$ we define
\[ \sin(\angle a) = \frac{2t_a}{t_a^2 + 1} \enspace ,\]
which, when choosing the appropriate root, gives
\[ t_a = \frac{1}{\sin(\angle a)} - \sqrt{\frac{1}{\sin^2(\angle a)} - 1} \enspace .\]
The sine of $\angle a$ is given by $a_y/\sqrt{a_x^2 + a_y^2}$.
The values of $\angle b$ and $t_b$ are defined analogously.
We therefore need to find a rational number $t$ with $t_a \leq t \leq t_b$.
The Cauchy bound (see~\cite{yap}) for an algebraic number $g$ being the root of a polynomial $\sum_{i=0}^m c_i x^i$ with rational coefficients $c_i$ is given by
\[ |g| \geq \frac{|c_0|}{|c_0| + \max\{ |c_1|,\dots,|c_m| \}} \enspace .\]
The difference $|t_a - t_b|$ is therefore bounded from below by a rational that has a denominator polynomial in the problem size.
{Using Farey approximation, we can find a rational $t$ whose denominator exceeds the denominator of the bound only by a polynomial factor.}

Since we can choose rational points on the unit circle inside an interval (and therefore on instances of the unit circle that are translated and scaled by rational values), we now have the tools to choose the endpoints of the tunnels.
For two vertices $v$ and $w$, let these be called $p_v$ and $q_v$ (placed on~$C_v$), as well as $p_w$ and $q_w$ (placed on~$C_w$).
Hence, a tunnel between $v$ and $w$ consists of the quadrilateral $p_v q_v q_w p_w$.
In order to prevent collinear triples of points, we again select a set of candidate points on $C_v$ and $C_w$ and choose the four points among them.
Note that this results in tunnels that may not be exactly rectangular, but this is irrelevant for our final construction.
See \figurename~\ref{fig_app_tunnel_endpoints} for an accompanying illustration.

We place the points in the following way.
Without loss of generality, suppose that %
$x_v < x_w$.
To obtain the set of candidate points for~$p_v$, consider the segment between $w$ and the point $(x_w, y_w + r_V)$, where $r_V$ is the radius of the circles, which we call the \emph{upper spoke} of $w$.
Let the \emph{lower spoke} of $w$ be defined analogously between $v$ and the point $(x_w, y_w\!-\!r_V)$.
Find the two parameters $t_1$ and $t_2$ for rational points $p_{t_1}$ and $p_{t_2}$ on $C_v$ such that the line through $v$ and $p_{t_1}$ intersects the upper spoke of $w$ at a point above $(x_w, y_w + 7 r_V/8)$ and the line through $v$ and $p_{t_2}$ intersects the upper spoke between $(x_w, y_w + r_V/2)$ and $(x_w, y_w + 5 r_V/8)$.
We can now select our set $K_v$ of candidate points from the interval $[t_1, t_2]$.
The same can be done for two parameters $t_3$ and $t_4$, with the roles of $v$ and $w$ interchanged.
We select a point $p_v \in K_v$ and a point $p_w \in K_w$ as the endpoints of one side of the tunnel gadget between $v$ and $w$.

Let us now argue the correctness of this construction.
Note that we do not need to require the sides to be parallel to the supporting line of $v$ and $w$ (we could do so by increasing the number of candidate points).
The crucial property of the points we need is that $p_v v w p_w$ forms a convex quadrilateral and we therefore have to prove that $p_v$ is always left of the directed line through $v$ and $p_w$ (and, analogously, that $p_w$ is right of the directed line through $w$ and $p_v$).
Let $l_w = (x_w, y_w + r_V/2)$ and $u_v = (x_v, y_v + r_V)$.
The diagonals $v l_w$ and $u_v w$ of the trapezoid $u_v v w l_w$ intersect each other at a ratio of $(r_V/2)/ r_V$, i.e., at two thirds of the interval $[x_v, x_w]$.
Recall that the radius $r_V$ was chosen in a way that the disc centers have a horizontal distance of at least $6r_V$.
Hence, the segments intersect outside~$C_w$;
the topmost candidate point on $C_w$ is below the line through $v$ and the lowest candidate point on $C_v$, and vice versa.
Note that since the candidate points on $C_v$ are chosen inside the convex hull of $C_w$ and $v$, no two tunnels from $v$ can intersect.

\begin{figure}
\centering
\includegraphics{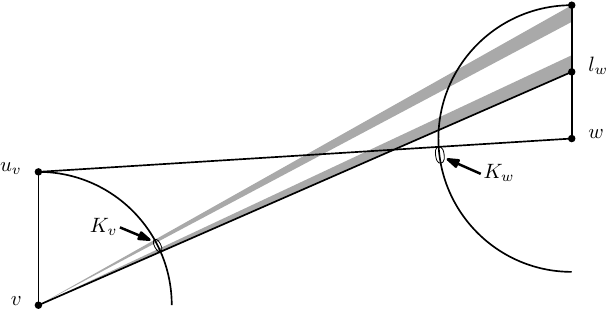}
\caption{Construction for the tunnel endpoints: The two extremal candidate points for $p_v$ are chosen inside the two gray wedges.
Note that $v$ and $w$ in the drawing do not fulfill the required vertical distance since the drawing would get too small.
}
\label{fig_app_tunnel_endpoints}
\end{figure}

It remains to find the correct number of candidate points.
Suppose we already constructed all but one tunnel point.
Since at every circle there are at most $2(n-1)$ tunnel points there are at most $\binom{n(2n-1)}{2}$ lines on which we are not allowed to place a point.
Every line intersects the circle on which we place the last point at most twice.
Hence, if we choose more than twice the number of points as we have lines, we can always choose a point such that the resulting point set is in general position.

\subsection{Points in the Tunnels}
For each tunnel, we construct two circular arcs, one for each segment defining the tunnel, on which we place the points of the edge core.
The crucial property of such an arc is that for two wire centers $v$ and $w$, these points are the only ones in the flip-kernel.
Let $q_v$ and $q_w$ be the two endpoints of a tunnel edge, s.t.\ $q_w$ is to the right of $q_v$ and the interior of the tunnel is above the line $q_v q_w$.
See \figurename~\ref{fig_app_center_points}.
The constructed arc will start at~$q_w$ and end at~$q_v$.
We consider four rays, namely the ones that leave $q_v$ to the right in an angle of $0, (-\pi/4)$, and $\pi/4$ with the $x$-axis and the upward vertical ray at $q_v$.
Let $r$ be the one that opens the smallest positive angle $\alpha$ with $q_v q_w$.
If the angle between $q_v$ and $q_w$ and $q_v w$ is smaller than $\alpha$, then let $s$ be the ray through $w$ starting at $q_v$; otherwise, let $s=r$.
We perform the analogous operation (i.e., with the plane being mirrored horizontally) at $q_w$, obtaining a ray~$s'$.
Without loss of generality, let the angle between $q_v q_w$ and $s$ be smaller than or equal to the one between $q_w q_v$ and $s'$.
Construct the circle $A$ that passes through both $q_v$ and~$q_w$ such that~$A$ is tangent to the supporting line of~$s$.
The coordinates of the center of $A$ are still rational.
It is well-known that, when given any rational point~$p$ on $A$ and a line $\ell$ with rational slope that intersects~$A$ at $p$ and a second point~$p'$, the point~$p'$ is rational as well, see, e.g.,~\cite[p.~5]{husemoeller}.
Hence, we need to appropriately choose lines through a point~$p$.
The crossings of the segments that define all the tunnels identify the region where the edge core should be placed.
Let $R$ be the region we have to place the points in (marked gray in \figurename~\ref{fig_app_center_points}).
By the choice of $r$, we constructed~$A$ in a way that we can mirror and rotate the plane orthogonally such that the intersection of $A$ and $R$ is within an angle of $0$ and $\pi/4$ from~$q_v$.
This means that any line $\ell$ through $q_v$ and this intersection will have a slope~$t$ between 0 and 1.
This reasoning is similar to the one of Burnikel~\cite{burnikel} to adapt the techniques of~\cite{canny} for such rational circles (i.e., circles given by three rational points).
As before, we can use, e.g., Farey approximation for the slope $t$ of $\ell$.
At each iteration, we check whether the second intersection of $\ell$ with $A$ is inside the quadrilateral $R$, and, if not, on which side it is.
Since the denominators of the coordinates of the points defining $A$ and $R$ are polynomial, there is a polynomial lower bound on the difference between the (possibly non-rational) parameters for the two points where $A$ enters and leaves $R$ (as for the construction of the tunnel endpoints).
Hence, after a polynomial number of steps, we have a rational slope for $\ell$ such that $\ell$ passes through $A$ inside $R$; therefore, also this intersection point has rational coordinates and its denominator is polynomial in the problem size.
To obtain a second such point, the process can be continued.
Now we have two points in the intersection of $A$ and $R$ which define two slopes of lines through $q_v$.
Any line through $q_v$ with a slope in the interval between these two slopes gives a rational point on $A \cap R$.
Hence, we can choose our candidate points by dividing that interval.

\begin{figure}
\centering
\includegraphics{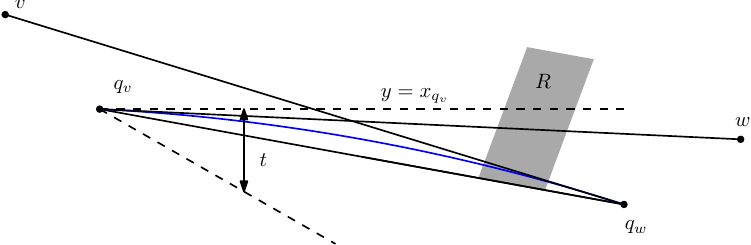}
\caption{We want to choose rational points on the (blue) arc inside the gray region by Farey approximation on the slope $t$.
Note that the gray region actually is, for presentational reasons, drawn too close to $w$.}
\label{fig_app_center_points}
\end{figure}

The remaining problem is the one of choosing the points for the crossing gadgets.
Two arcs in crossing tunnels will, in general, cross at a point that does not have rational coordinates.
The crucial property of the points of the crossing gadgets, however, is that they are outside the hourglasses of the edge cores (recall Definition~\ref{def_dc}) and that the edges between them can ``quickly'' be flipped to the center of the corresponding wiring gadget.
Placing the points for the crossing gadgets at the crossings of the segments that define the tunnels would satisfy these constraints, but would lead to collinear triples.
So we have to slightly perturb each point $p$ to obtain a point~$p'$ without loosing these properties.
See \figurename~\ref{fig_app_tunnel_crossings}.
Between every consecutive pair of crossing points on a tunnel segment $q_v q_w$ we can choose the rational midpoint.
If the perturbed point~$p'$ remains on the same side of the line through the wire center and the midpoint as $p$, the order around the wire center is maintained.
Further, the perturbed points have to remain on the same sides of the lines that define the hourglasses of the edge cores involved.
Together with the tunnel edges, these constraints give a convex region from which we can choose our perturbed point.
We may again place a circular arc inside this region (marked gray in \figurename~\ref{fig_app_tunnel_crossings}) on which we select a sufficiently large number of candidate points, analogously to the construction of the other gadgets.

\begin{figure}
\centering
\includegraphics{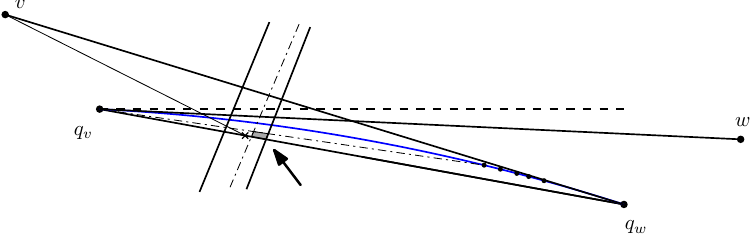}
\caption{Construction of tunnel crossings. The drawing shows the lower part of a tunnel between $v$ and $w$ and a part of another tunnel (indicated by the near-vertical strokes).}
\label{fig_app_tunnel_crossings}
\end{figure}

\subsection{Points for the Wiring}
Finally, we place the points at the wiring gadgets that allow us to draw the wiring edges, see \figurename~\ref{fig_app_wiring_points}.
The circles for the wiring gadgets are scaled versions of the unit circle.
For a vertex~$v$, let $\ell_{v}$ be the line through $v$ that is perpendicular to the supporting line of the origin $o$ and $v$.
Since the coordinates of $v$ are rational sines, the intersection points of $\ell_{v}$ with the circle~$C_v$ are rational as well.
Due to the choice of $\delta_\mathrm{n}'$, all points on $C_v$ that define tunnels are on the same side of $\ell_{v}$ as~$o$.
We are given two intervals, each between two rational sines, i.e., between the ``extremal'' tunnel endpoints on~$C_v$ and the intersection points of $\ell_{v}$ with $C_v$.
Therefore, we can choose a sufficient number of rational candidate points on $C_v$ to choose the points for the wiring from.

\begin{figure}
\centering
\includegraphics{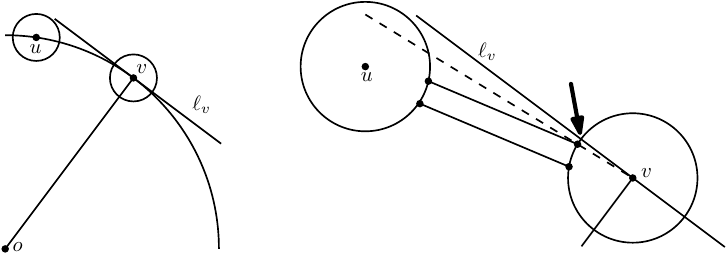}
\caption{Construction of the wiring points.
The small gap (indicated by the arrow) on the circle $C_v$ of $v$ between the intersection point with $\ell_{v}$ and the neighboring tunnel endpoint can be used for the candidate points.}
\label{fig_app_wiring_points}
\end{figure}

\subsection{Concluding Remarks on the Embedding}
The crucial part throughout the whole embedding procedure is that each (intermediate) point that is not a candidate point is constructed using only a constant number of other points.
The candidate points were constructed with polynomial parameters.
Hence, all denominators are polynomial in the input size.
In particular, note that even though some intervals were defined by points with algebraic coordinates, a lower bound on the interval can be given in terms of the other, rational coordinates that were used in the construction.
This allowed us to find rational points with polynomial denominators within these intervals.

\ifarxiv
\bibliographystyle{habbrv}
\fi
\bibliography{bibliography}

\end{document}